\RequirePackage{fix-cm}
\documentclass{svjour3}                     
\smartqed  
\usepackage{graphicx}
\usepackage{amsfonts,color,morefloats,pslatex}
\usepackage{amssymb, amsmath,latexsym}

\newtheorem{Open}{Open Problem}

\newcommand{\tr}{{\mathrm{Tr}}}

\newcommand{\gf}{{\mathrm{GF}}}

\newcommand{\Z}{\mathbb{{Z}}}
\newcommand{\ls}{{\mathbb{L}}}

\newcommand{\bC}{{\mathbb{C}}}

\newcommand{\m}{\mathbb{M}}

\newcommand{\cS}{{\mathcal{S}}}

\newcommand{\C}{{\mathcal{C}}}

\begin{document}

\title{Cyclic Codes from Dickson Polynomials\thanks{C. Ding's research was supported by the Hong Kong Research Grants Council, Proj. No. 16300415.}
}


\author{Cunsheng Ding
}


\institute{Cunsheng Ding \at
              Department of Computer Science and Engineering, The Hong Kong University of Science and Technology, Clear Water Bay, Kowloon, Hong Kong, China \\
              Tel.: +852-2358 7021, \ 
              Fax: +852-2358 1477, \ 
              \email{cding@ust.hk}           
}

\date{Received: date / Accepted: date}

\maketitle

\begin{abstract}
Due to their efficient encoding and decoding algorithms, cyclic codes, a subclass 
of linear codes, have applications in communication systems, consumer electronics, 
and data storage systems. In this paper, Dickson polynomials of the first 
kind over finite fields are employed to construct a number of classes of 
cyclic codes. Lower bounds on the minimum weight of some classes of the cyclic 
codes are developed. The minimum weights of some other classes of the codes 
constructed in this paper are determined. The dimensions of the codes obtained 
in this paper are flexible. Many of the codes presented  in this paper are optimal 
or almost optimal. 

\end{abstract}

\section{Introduction}

Let $q$ be a power of a prime $p$. 
An $[n,k, d]$ linear code over $\gf(q)$ is a $k$-dimensional subspace of $\gf(q)^n$ 
with minimum (Hamming) nonzero weight $d$. 
Let $A_i$ denote the number of codewords with Hamming weight $i$ in a linear code
$\C$ of length $n$. The {\em weight enumerator} of $\C$ is defined by
$$
1+A_1z+A_2z^2+ \cdots + A_nz^n.
$$
The {\em weight distribution} of $\C$ is the sequence $(1,A_1,\ldots,A_n)$.

An $[n, k, d]$ linear code over $\gf(q)$ is called \textit{optimal} if there is no 
$[n, k, d+1]$ or $[n, k+1, d]$ linear code over $\gf(q)$. The optimality of a cyclic 
code may be proved by a bound on linear codes or by an exhaustive computer search on 
all linear codes over $\gf(q)$ with fixed length $n$ and fixed dimension $k$ or fixed 
length $n$ and fixed minimum distance $d$. An 
$[n, k, d]$ linear code is said to be \textit{almost optimal} if a linear code with 
parameters $[n, k+1, d]$ or $[n, k, d+1]$ is optimal.   

A vector $(c_0, c_1, \cdots, c_{n-1}) \in \gf(q)^n$ is said to be {\em even-like} 
if $\sum_{i=0}^{n-1} c_i =0$, and is {\em odd-like} otherwise. The {\em even-like subcode} of a 
linear code consists of all the even-like codewords of this linear code.

An $[n,k]$ linear code $\C$ over $\gf(q)$ is called {\em cyclic} if 
$(c_0,c_1, \cdots, c_{n-1}) \in \C$ implies $(c_{n-1}, c_0, c_1, \cdots, c_{n-2}) 
\in \C$.  
Let $\gcd(n, q)=1$. By identifying any vector $(c_0,c_1, \cdots, c_{n-1}) \in \gf(q)^n$ 
with  
$$ 
c_0+c_1x+c_2x^2+ \cdots + c_{n-1}x^{n-1} \in \gf(q)[x]/(x^n-1), 
$$
any code $\C$ of length $n$ over $\gf(q)$ corresponds to a subset of $\gf(q)[x]/(x^n-1)$. 
The linear code $\C$ is cyclic if and only if the corresponding subset in $\gf(q)[x]/(x^n-1)$ 
is an ideal of the ring $\gf(q)[x]/(x^n-1)$. 
It is well known that every ideal of $\gf(q)[x]/(x^n-1)$ is principal. Let $\C=(g(x))$ be a 
cyclic code, where $g$ is monic and has the least degree. Then $g(x)$ is called the {\em generator polynomial} and 
$h(x)=(x^n-1)/g(x)$ is referred to as the {\em check} polynomial of 
$\C$.  

The error correcting capability of cyclic codes may not be as good as some other linear 
codes in general. However, cyclic codes have wide applications in storage and communication 
systems because they have efficient encoding and decoding algorithms 
\cite{Chie,Forn,Pran}.

Cyclic codes have been studied for decades and a lot of  progress has been made 
(see, for example, \cite{Char,HPbook} for information).  The total number of cyclic codes 
over $\gf(q)$ and their constructions are closely related to cyclotomic cosets 
modulo $n$, and thus many areas of number theory. One way of   
constructing cyclic codes over $\gf(q)$ with length $n$ is  to use the generator polynomial 
\begin{eqnarray}\label{eqn-defseqcode}
\frac{x^n-1}{\gcd(S(x), x^n-1)}
\end{eqnarray}
where 
$$ 
S(x)=\sum_{i=0}^{n-1} s_i x^i  \in \gf(q)[x]   
$$
and $s^{\infty}=(s_i)_{i=0}^{\infty}$ is a sequence of period $n$ over $\gf(q)$. 
Throughout this paper, we call the cyclic code $\C_s$ with the generator polynomial 
of (\ref{eqn-defseqcode}) the {\em code defined by the sequence} $s^{\infty}$, 
and the sequence $s^{\infty}$ the {\em defining sequence} of the cyclic code $\C_s$. 
This approach was successfully employed to construct cyclic codes with interesting 
parameters in \cite{Ding120,Ding121,Ding13,DingBK15,DZ2014,TQX,Weld}.

In this paper, Dickson polynomials of the first kind and small degrees over finite 
fields will be employed to construct 
a number of classes of sequences and then cyclic codes. Lower bounds on the minimum weight of some 
classes of the cyclic codes are developed. The minimum weights of some other 
classes of the codes constructed in this paper are determined. The dimensions 
of the codes of this paper are flexible.  It is amazing that most of the cyclic codes 
from Dickson polynomials of the first kind with small degrees are optimal or almost 
optimal. The major motivation of this paper is 
the optimality of many of these cyclic codes from Dickson polynomials. 
Another motivation of this study is the simplicity of the constructions of the cyclic 
codes in this paper. 

\section{Preliminaries} 

In this section, we present basic notations and results of Dickson polynomials, 
$q$-cyclotomic cosets, and sequences that will be employed in subsequent sections.     

\subsection{Some notation and symbols fixed throughout this paper}\label{sec-notations} 

Throughout this paper, we adopt the following notation unless otherwise stated: 
\begin{itemize} 
\item $p$ is a prime, $q$ is a positive power of $p$,  $m$ is a positive integer, $r=q^m$, and $n=q^m-1$.  
\item $\Z_n=\{0,1,\cdots, n-1\}$, the ring of integers modulo $n$.
\item $\alpha$ is a generator of $\gf(r)^*$, the multiplicative group of $\gf(q)$.  
\item $m_a(x)$ is the minimal polynomial of $a \in \gf(r)$ over $\gf(q)$.
\item $\tr(x)$ is the trace function from $\gf(r)$ to $\gf(q)$.   
\item $\delta(x)$ is a function on $\gf(r)$ defined by $\delta(x) =0$ if $\tr(x)=0$ and $\delta(x) =1$ otherwise.   
\item For any polynomial $g(x) \in \gf(q)[x]$ with $g(0) \ne 0$, $\bar{g}(x)$ denotes the reciprocal of 
         $g(x)$. 
\item For any code $\C$ over $\gf(q)$ with generator polynomial $g(x)$, $\bar{\C}$ denotes the cyclic 
         code with generator polynomial $\bar{g}(x)$. It is well known that $\C$ and $\bar{\C}$ have the 
         same weight distribution.       
\end{itemize}

\subsection{The $q$-cyclotomic cosets modulo $n=q^m-1$}\label{sec-cpsets}

The $q$-cyclotomic coset containing $j$ modulo $n$ is defined by 
$$ 
C_j=\{j, qj, q^2j, \cdots, q^{\ell_j-1}j\} \subset \Z_n
$$
where $\ell_j$ is the smallest positive integer such that $q^{\ell_j}j \equiv j \pmod{n}$, 
and is called the size of $C_j$. It is known that $\ell_j$ divides $m$. The smallest integer 
in $C_j$ is called the {\em coset leader} of $C_j$. Let $\Gamma$ denote the set of all 
coset leaders. By definition, we have 
$$ 
\bigcup_{j \in \Gamma} C_j =\Z_n.  
$$ 
It is easily seen that $\ell_i=\ell_{n-i}$ for all $i$. 

It is well known that $\prod_{j \in C_i} (x-\alpha^j)$ is an irreducible polynomial of degree 
$\ell_i$ over $\gf(q)$ and is the minimal polynomial of $\alpha^i$ over $\gf(q)$. Furthermore,  
the canonical factorization of $x^n-1$ over $\gf(q)$ is given by   
$$ 
x^n-1=\prod_{i \in \Gamma} \prod_{j \in C_i} (x-\alpha^j).  
$$

\subsection{The linear span and minimal polynomial of sequences}

Let $s^L=s_0s_1\cdots s_{L-1}$ be a sequence over $\gf(q)$. The {\em linear 
span} (also called {\em linear complexity}) of $s^L$ is defined to be the smallest positive 
integer $\ell$ such that there are constants $c_0=1, c_1, \cdots, c_\ell \in \gf(q)$ 
satisfying 
\begin{eqnarray*} 
-c_0s_i=c_1s_{i-1}+c_2s_{i-2}+\cdots +c_ls_{i-\ell} \mbox{ for all } \ell \leq i<L. 
\end{eqnarray*} 
In engineering terms, such a polynomial $c(x)=c_0+c_1x+\cdots +c_lx^l$ 
is called the {\em feedback  polynomial} of a shortest linear feedback 
shift register 
(LFSR) that generates $s^L$. Such an integer always exists for finite sequences  $s^L$. When 
$L$ is $\infty$, a sequence $s^{\infty}$ is called a semi-infinite 
sequence. If there is no such an integer for a semi-infinite sequence 
$s^{\infty}$, its linear span is defined to be $\infty$. The linear 
span of the zero sequence is defined to be zero. 
For ultimately periodic semi-infinite sequences such an $\ell$ always 
exists. 

Let $s^{\infty}$ be a sequence of period $L$ over $\gf(q)$. 
Any feedback polynomial of $s^{\infty}$ is called a {\em characteristic 
polynomial}. The characteristic polynomial with the smallest degree is 
called the {\em minimal polynomial} of the periodic sequence $s^{\infty}$. 
Since we require that the constant term of any characteristic polynomial 
be 1, the minimal polynomial of any periodic sequence $s^{\infty}$ must 
be unique. In addition, any characteristic polynomial must be a multiple 
of the minimal polynomial.    

For periodic sequences, there are a few ways to determine their linear 
span and minimal polynomials. One of them is given in \cite[Theorem 5.3]{DXS}. 
The other one is given in the following lemma \cite{Antweiler} 

\begin{lemma} \label{lem-ls2} 
Any sequence $s^{\infty}$ over $\gf(q)$ of period $q^m-1$ has a unique expansion of the form  
\begin{equation*}
s_t=\sum_{i=0}^{q^m-2}c_{i}\alpha^{it}, \mbox{ for all } t\ge 0,
\end{equation*}
where $\alpha$ is a generator of $\gf(q^m)^*$ and $c_i \in \gf(q^m)$.
Let the index set $I=\{i \left.\right| c_i\neq 0\}$, then the minimal polynomial $\m_s(x)$ of $s^{\infty}$ is 
\begin{equation*}
\m_s(x)=\prod_{i\in I}(1-\alpha^i x),
\end{equation*}
and the linear span of $s^{\infty}$ is $|I|$.
\end{lemma}

It should be noticed that in some references the reciprocal of $\m_s(x)$ is called the minimal polynomial 
of the sequence $s^\infty$. So Lemma \ref{lem-ls2} is a modified version of the original one in \cite{Antweiler}.

\subsection{Dickson polynomials over $\gf(r)$}\label{sec-DPAPNPN} 

One hundred and sixteen years ago, Dickson introduced the following family of polynomials over 
$\gf(r)$ \cite{Dick96}: 
\begin{eqnarray}\label{eqn-1stDP}
D_h(x, a)=\sum_{i=0}^{\lfloor \frac{h}{2} \rfloor} \frac{h}{h-i} \binom{h-i}{i}  (-a)^i x^{h-2i}, 
\end{eqnarray} 
where $a \in \gf(r)$ and $h \ge 0$ is called the {\em order} of the polynomial. This family is 
referred to as the {\em Dickson polynomials of the first kind}.

Dickson polynomials of the second kind over $\gf(r)$ are defined by  
\begin{eqnarray}\label{eqn-2ndDP}
E_h(x, a)=\sum_{i=0}^{\lfloor \frac{h}{2} \rfloor} \binom{h}{h-i} (-a)^i x^{h-2i}, 
\end{eqnarray} 
where $a \in \gf(r)$ and $h \ge 0$ is called the {\em order} of the polynomial.

Dickson polynomials are an interesting topic of mathematics and engineering, and have many applications. 
For example, the Dickson polynomials $D_5(x, a)=x^5-ux-u^2x$ over $\gf(3^m)$ are 
employed to construct a family of planar functions \cite{CM,DY06}, and those planar 
functions give two families of commutative presemifields, planes,  several classes of 
linear codes \cite{CDY,YCD}, and two families of skew Hadamard difference sets \cite{DY06}. 
The reader is referred to \cite{LMT} for detailed information about Dickson polynomials. 
In this paper, we will employ Dickson polynomials of the first kind over finite fields to 
construct cyclic codes with some interesting parameters.

\section{The construction of cyclic codes from polynomials over $\gf(r)$}\label{sec-polycode} 

Given a polynomial $f(x)$ on $\gf(r)$, we define its associated sequence 
$s^\infty$ by 
\begin{eqnarray}\label{eqn-sequence}
s_i=\tr(f(\alpha^i+1)) 
\end{eqnarray}
for all $i \ge 0$, where $\alpha$ is a generator of $\gf(r)^*$ and $\tr(x)$ denotes 
the trace function from $\gf(r)$ to $\gf(q)$. 

It was demonstrated in \cite{Ding13,DZ2014,TQX} that the code $\C_s$ may have interesting 
parameters if the polynomial $f$ is properly chosen.  
The objective of this paper is to consider cyclic codes $\C_s$ defined 
by Dickson polynomials $f$ over $\gf(r)$ with small degrees. 


\section{Cyclic codes from the Dickson polynomial $D_{p^u}(x,a)$}\label{sec-tracex}

Since $q$ is a power of $p$, it is known that $D_{hp}(x,a)=D_{h}(x,a)^p$ \cite[Lemma 2.6 ]{LMT}.  
It then follows that 
$$ 
D_{p^u}(x,a)=x^{p^u}
$$
for all $a \in \gf(r)$. 

The code $\C_s$ over $\gf(q)$ defined by the Dickson polynomial $f(x)=D_{p^u}(x,a)=x^{p^u}$ 
over $\gf(q^m)$ are not new. However, for the completeness of cyclic codes from Dickson 
polynomials we state the following theorem without giving a proof.     

\begin{theorem}\label{thm-tracex} 
The code $\C_{s}$ defined by the Dickson polynomial  $D_{p^u}(x,a)=x^{p^u}$ has parameters 
$[n, n-m-\delta(1), d]$ and generator polynomial $\m_s(x)=(x-1)^{\delta(1)} m_{\alpha^{-p^u}}(x)$,    
where 
\begin{eqnarray*} 
\left\{ \begin{array}{l} 
d=4 \mbox{ if }  q=2 \mbox{ and } \delta(1)=1, \\
d=3 \mbox{ if }  q=2 \mbox{ and } \delta(1)=0, \\
d=3 \mbox{ if }  q>2 \mbox{ and } \delta(1)=1, \\
d=2 \mbox{ if }  q>2 \mbox{ and } \delta(1)=0,  
 \end{array} 
 \right. 
 \end{eqnarray*} 
where the function $\delta(x)$ and the polynomial $m_{\alpha^j}(x)$ were defined in Section \ref{sec-notations}. 
\end{theorem}

When $q=2$, the code of Theorem \ref{thm-tracex} is equivalent to the binary Hamming weight or 
its even-weight subcode, and is thus optimal. The code is either optimal or almost optimal with 
respect to the Sphere Packing Bound.

\section{Cyclic codes from $D_2(x,a)=x^2-2a$}\label{} 

In this section we consider the code $\C_s$ defined by  $f(x)=D_2(x,a)=x^2-2a$ 
over $\gf(r)$. When $p=2$, this code was treated in Section \ref{sec-tracex}. When 
$p>2$, the following theorem is a variant of Theorem 5.2 in \cite{Ding121}, but 
has much stronger conclusions on the minimum distance of the code.  

\begin{theorem}\label{thm-DP2square} 
Let $p>2$ and $m \geq 3$. 
The code $\C_{s}$ defined by $f(x)=D_2(x,a)=x^2-2a$ has parameters 
$[n, n-2m-\delta(1-2a), d]$ and generator polynomial 
\begin{equation*} 
\m_s(x)= (x-1)^{\delta(1-2a)} m_{\alpha^{-1}}(x) m_{\alpha^{-2}}(x), 
\end{equation*} 
where 
\begin{eqnarray*} 
d=\left\{ \begin{array}{ll} 
4   & \mbox{ if }  q=3 \mbox{ and } \delta(1-2a)=0, \\
5  & \mbox{ if }  q=3 \mbox{ and } \delta(1-2a)=1, \\
3   & \mbox{ if }  q>3 \mbox{ and } \delta(1-2a)=0, \\
4 & \mbox{ if }  q>3 \mbox{ and } \delta(1-2a)=1,  
 \end{array} 
 \right. 
 \end{eqnarray*} 
 and the function $\delta(x)$ and the polynomial $\m_{\alpha^j}(x)$ were defined in Section \ref{sec-notations}. 
\end{theorem} 

\begin{proof}
The conclusion on the minimal polynomial $\m_s(x)$ was proved in Lemma 5.1 in \cite{Ding121}. 
Hence, the conclusion on the dimension of $\C_s$ then follows. It remains to determine the 
minimum distance of the code.  

When $\delta(1-2a)=1$, the desired conclusions on $d$ were proved in \cite{CDY}. 
When $\delta(1-2a)=0$, the desired conclusions on $d$ can be proved similarly by 
modifying the proof of Theorem 7 in \cite{CDY}.    
\end{proof}

The code of Theorem \ref{thm-DP2square} is either optimal or almost optimal for all $m \geq 2$. 
We now prove this statement as follows. 
\begin{itemize}
\item When $q=3$ and $\delta(1-2a)=0$, $\C_s$ has parameters $[3^m-1, 3^m-1-2m, 4]$. The Sphere 
      Packing Bound shows that there is no linear code over $\gf(3)$ with parameters 
      $[3^m-1, 3^m-1-2m, 5]$. By definition, $\C_s$  is optimal. 
\item When $q=3$ and $\delta(1-2a)=1$, $\C_s$ has parameters $[3^m-1, 3^m-1-2m-1, 5]$. The Sphere 
      Packing Bound shows that there is no linear code over $\gf(3)$ with parameters 
      $[3^m-1, 3^m-1-2m, 5]$. By definition, $\C_s$  is optimal.       
\item When $q>3$ and $\delta(1-2a)=0$, $\C_s$ has parameters $[q^m-1, q^m-1-2m, 3]$ and is thus almost 
      optimal, as the Sphere Packing Bound shows that any linear code over $\gf(q)$ with
       parameters 
      $[q^m-1, q^m-2m, 4]$ is optimal.   
\item When $q>3$ and $\delta(1-2a)=1$, $\C_s$ has parameters $[q^m-1, q^m-1-2m-1, 4]$. The Sphere 
      Packing Bound shows that there is no linear code over $\gf(q)$ with parameters 
      $[q^m-1, q^m-1-2m-1, 5]$. By definition, $\C_s$  is optimal.         
\end{itemize}

\section{Cyclic codes from $D_3(x,a)=x^3-3ax$}\label{sec-order3} 

In this section we study the code $\C_s$ defined by the Dickson polynomial 
$D_3(x,a)=x^3-3ax$. We need to distinguish among the three cases: $p=2$, 
$p=3$ and $p \ge 5$. The case that $p=3$ was covered in Section \ref{sec-tracex}. 
So we need to consider only the two remaining cases.  

We first handle the case $q=p=2$ and  prove the following lemma.  

\begin{lemma}\label{lem-2DPDOGold} 
Let $q=p=2$. 
Let $s^{\infty}$ be the sequence of (\ref{eqn-sequence}), where $f(x)=D_3(x,a)=x^3-3ax=x^3+ax$.  
Then the minimal polynomial $\m_s(x)$ of  $s^{\infty}$ is given by 
\begin{eqnarray*}
\m_s(x)=\left\{ 
\begin{array}{ll} 
(x-1)^{\delta(1)} m_{\alpha^{-3}}(x)  & \mbox{if } a =0, \\
(x-1)^{\delta(1+a)} m_{\alpha^{-1}}(x) m_{\alpha^{-3}}(x) & \mbox{if } a \ne 0 
\end{array} 
\right. 
\end{eqnarray*} 
where $m_{\alpha^{-j}}(x)$ and the function $\delta(x)$ were defined in Section \ref{sec-notations}, 
and the linear span $\ls_s$ of $s^{\infty}$ is given by 
\begin{eqnarray*}
\ls_s=\left\{ 
\begin{array}{ll} 
\delta(1) + m   & \mbox{if } a =0, \\
\delta(1+a) + 2m & \mbox{if } a \ne 0.  
\end{array} 
\right. 
\end{eqnarray*} 
\end{lemma} 

\begin{proof} 
Note that 
$$ 
D_3(x+1, a)=x^3+x^2+(1+a)x+1+a. 
$$
We have then 
$$ 
\tr(D_3(x+1, a))=\tr(x^3+ax)+\tr(1+a). 
$$
By definition, 
\begin{eqnarray}\label{eqn-2DPDOGold2}
s_t = \tr((\alpha^t)^3+a\alpha^t)+\tr(1+a).
\end{eqnarray}

It can be easily proved that $\ell_{1}=\ell_{n-1}=\ell_{3}=\ell_{n-3}=m$ 
and that $C_1 \cap C_{3}=\emptyset$. 
The desired conclusions on the linear span and the minimal polynomial $\m_s(x)$ then follow from Lemma \ref{lem-ls2} 
and (\ref{eqn-2DPDOGold2}). 
\end{proof}

The following theorem gives information on the code $\C_{s}$.    

\begin{theorem}\label{thm-2DO2Gold} 
Let $q=p=2$ and let $m \geq 4$. 
Then the binary code $\C_{s}$ defined by the sequence of Lemma \ref{lem-2DPDOGold} has parameters 
$[n, n-\ls_s, d]$ and generator polynomial $\m_s(x)$, where $\m_s(x)$ and $\ls_s$ are 
given in Lemma \ref{lem-2DPDOGold}, and 
\begin{eqnarray*}
d=\left\{ 
\begin{array}{ll} 
2   & \mbox{if } a =0 \mbox{ and }  \delta(1)=0, \\
4   & \mbox{if } a =0 \mbox{ and }  \delta(1)=1, \\
5   & \mbox{if } a \ne 0 \mbox{ and }  \delta(1+a)=0, \\
6   & \mbox{if } a \ne 0 \mbox{ and }  \delta(1+a)=1. 
\end{array} 
\right. 
\end{eqnarray*} 
\end{theorem} 

\begin{proof} 
The dimension of $\C_{s}$ follows from Lemma \ref{lem-2DPDOGold} and the definition of the 
code $\C_s$.  We need to prove the conclusion on the minimum distance $d$ of $\C_{s}$. 

We consider the case $a=0$ first. Since $\alpha^3 \ne 0$, $d \ge 2$. On the other hand, 
if $\delta(1)=0$, then $m$ is even and $(\alpha^{3})^{(2^m-1)/3}=1$. Hence $\C_s$ has 
a codeword of Hamming weight 2. Whence, $d=2$. If $\delta(1)=1$, then $m$ is odd and 
$\gcd(3, 2^m-1)=1$. Hence, $\alpha^3$ is a primitive element of $\gf(2^m)$ and the code 
$\tilde{\C}_s$ generated by $\m_{\alpha^{-3}}(x)$ is equivalent to the binary Hamming 
code, and has thus minimum weight 3. Hence the even-weight 
subcode $\C_s$ of  $\tilde{\C}_s$ has minimum weight 4. 

We now consider the case that $a \ne 0$. 
When $\delta(1+a)=1$, it was proved in \cite{DFZ16} that $d=5$. When $\delta(1+a)=0$,  
the code is the even-like subcode of the code in the case $\delta(1+a)=1$. In this case,  
$d=6$.   
\end{proof} 

\begin{remark} 
When $a=0$ and $\delta(1)=1$, the code is equivalent to the even-weight subcode of 
the Hamming code. We are mainly interested in the case that $a \ne 0$. When $a=1$, the code 
$\C_s$ is a double-error correcting binary BCH code or its even-like  
subcode. Theorem \ref{thm-2DO2Gold} shows that well-known classes of cyclic codes can be 
constructed with Dickson polynomials of order 3. The code is either optimal or almost optimal. 
\end{remark}

Now we investigate the case $q=p^t$, where $p \ge 5$ or $p=2$ and $t\ge 2$. 

\begin{lemma}\label{lem-3tDPDOGold} 
Let $q=p^t$, where $p \ge 5$ or $p=2$ and $t\ge 2$. 
Let $s^{\infty}$ be the sequence of (\ref{eqn-sequence}), where $f(x)=D_3(x,a)=x^3-3ax$.  
Then the minimal polynomial $\m_s(x)$ of  $s^{\infty}$ is given by 
\begin{eqnarray*}
\m_s(x)=\left\{ 
\begin{array}{ll} 
(x-1)^{\delta(-2)} m_{\alpha^{-3}}(x)  m_{\alpha^{-2}}(x) & \mbox{if } a =1, \\
(x-1)^{\delta(1-3a)} m_{\alpha^{-3}}(x) m_{\alpha^{-2}}(x)  m_{\alpha^{-1}}(x) & \mbox{if } a \ne 1 
\end{array} 
\right. 
\end{eqnarray*} 
where $m_{\alpha^{-j}}(x)$ and the function $\delta(x)$ were defined in Section \ref{sec-notations}, 
and the linear span $\ls_s$ of $s^{\infty}$ is given by 
\begin{eqnarray*}
\ls_s=\left\{ 
\begin{array}{ll} 
\delta(-2) + 2m   & \mbox{if } a =1, \\
\delta(1+a) + 3m & \mbox{if } a \ne 1.  
\end{array} 
\right. 
\end{eqnarray*} 
\end{lemma} 

\begin{proof} 
Note that 
$$ 
D_3(x+1, a)=x^3+3x^2+3(1-a)x+1-3a. 
$$
We have then 
\begin{eqnarray}\label{eqn-3tDPDOGold2}
s_t = \tr((\alpha^t)^3+3(\alpha^t)^2+3(1-a)\alpha^t)+\tr(1-3a).
\end{eqnarray}

Since $q=p^t$, where $p \ge 5$ or $p=2$ and $t\ge 2$, one can prove 
that $\ell_{1}=\ell_{n-1}=\ell_{3}=\ell_{n-3}=\ell_2=\ell_{n-2}=m$ 
and that 
$$ 
C_1 \cap C_{2}=\emptyset, \ C_1 \cap C_{3}=\emptyset,  \ C_2 \cap C_{3}=\emptyset.  
$$ 
The desired conclusions on the linear span and the minimal polynomial $\m_s(x)$ then follow from Lemma \ref{lem-ls2} 
and (\ref{eqn-3tDPDOGold2}). 
\end{proof}

The following theorem provides information on the code $\C_{s}$.    

\begin{theorem}\label{thm-3tDO2Gold} 
Let $q=p^t$, where $p \ge 5$ or $p=2$ and $t\ge 2$. 
Then the code $\C_{s}$ defined by the sequence of Lemma \ref{lem-3tDPDOGold} has parameters 
$[n, n-\ls_s, d]$ and generator polynomial $\m_s(x)$, where $\m_s(x)$ and $\ls_s$ are 
given in Lemma \ref{lem-3tDPDOGold}, and 
\begin{eqnarray*}
\left\{ 
\begin{array}{ll} 
d \ge 3   & \mbox{if } a =1, \\
d \ge 4   & \mbox{if } a \ne 1 \mbox{ and }  \delta(1-3a)=0, \\
d \ge 5   & \mbox{if } a \ne 1 \mbox{ and }  \delta(1-3a)=1,  \\  
d \ge 5   & \mbox{if } a \ne 1 \mbox{ and }  \delta(1-3a)=0 \mbox{ and } q=4, \\
d \ge 6   & \mbox{if } a \ne 1 \mbox{ and }  \delta(1-3a)=1 \mbox{ and } q=4.  
\end{array} 
\right. 
\end{eqnarray*} 
\end{theorem} 

\begin{proof} 
The dimension of $\C_{s}$ follows from Lemma \ref{lem-3tDPDOGold} and the definition of the 
code $\C_s$.  We now prove the conclusion on the minimum distance $d$ of $\C_{s}$. 

Note that $\bar{\m}_s(x)$ has the zeros $\alpha^2$ and  $\alpha^3$. By the BCH bound, 
$d \ge 3$ for all cases. If $a \ne 1$, $\bar{\m}_s(x)$ has the zeros $\alpha^i$ for all $i 
\in \{1,2,3\}$ and the additional zero $\alpha^0$ if $\delta(1-3a)=1$. Hence, the second and 
third lower bound on $d$ follow also from the BCH bound. 

The case $q=4$ is special. In this case, $\bar{\m}_s(x)$ has the zeros $\alpha^i$ for all $i 
\in \{1,2,3,4\}$ and the additional zero $\alpha^0$ if $\delta(1-3a)=1$. Hence, the last two lower  
bounds on $d$ also follow from the BCH bound. 
\end{proof} 

\begin{remark} 
The code $\C_{s}$ of Theorem \ref{thm-3tDO2Gold} is either a BCH code or the even-like subcode 
of a BCH code. 
One can similarly show that the code is either optimal 
or almost optimal. 

When $q=4$, $a \neq 1$, $\delta(1-3a)=1$, and $m \geq 3$, the Sphere Packing Bound shows that 
$d=6$. But the minimum distance is still open in other cases.  
\end{remark} 

\begin{Open} 
Determine the minimum distance $d$ for the code $\C_{s}$ of Theorem \ref{thm-3tDO2Gold}. 
\end{Open}

\section{Cyclic codes from $D_4(x,a)=x^4-4ax^2+2a^2$}\label{sec-4DPDOGold} 

In this section we investigate the code $\C_s$ defined by the Dickson polynomial 
$D_4(x,a)=x^4-4ax^2+2a^2$. We have to distinguish among the three cases: $p=2$, 
$p=3$ and $p \ge 5$. The case $p=2$ was covered in Section \ref{sec-tracex}. 
So we need to consider only the two remaining cases.  

We first take care of the cas $q=p=3$ and  prove the following lemma.  

\begin{lemma}\label{lem-34DPDOGold} 
Let $q=p=3$ and $m \ge 3$. 
Let $s^{\infty}$ be the sequence of (\ref{eqn-sequence}), where $f(x)=D_4(x,a)=x^4-4ax^2+2a^2$.  
Then the minimal polynomial $\m_s(x)$ of  $s^{\infty}$ is given by 
\begin{eqnarray*}
\m_s(x)=
 \left\{ 
\begin{array}{l} 
(x-1)^{\delta(1)} m_{\alpha^{-4}}(x)  m_{\alpha^{-1}}(x)  \mbox{ if } a =0, \\
(x-1)^{\delta(1)} m_{\alpha^{-4}}(x)  m_{\alpha^{-2}}(x)  \mbox{ if } a =1, \\
(x-1)^{\delta(1-a-a^2)} m_{\alpha^{-4}}(x)  m_{\alpha^{-2}}(x) m_{\alpha^{-1}}(x)  \mbox{ otherwise,}  
\end{array} 
\right. 
\end{eqnarray*} 
where $m_{\alpha^{-j}}(x)$ and the function $\delta(x)$ were defined in Section \ref{sec-notations}, 
and the linear span $\ls_s$ of $s^{\infty}$ is given by 
\begin{eqnarray*}\label{eqn-34DPDOGold}
\ls_s=\left\{ 
\begin{array}{ll} 
\delta(1) + 2m   & \mbox{if } a =0, \\
\delta(1) + 2m   & \mbox{if } a =1, \\
\delta(1-a-a^2) + 3m & \mbox{otherwise.}   
\end{array} 
\right. 
\end{eqnarray*} 
\end{lemma} 

\begin{proof} 
Note that 
$$ 
D_4(x+1, a)=x^4+x^3-ax^2+(1+a)x+1-a-a^2. 
$$
We have then 
$$ 
\tr(D_4(x+1, a))=\tr(x^4-ax^2+(a-1)x)+\tr(1-a-a^2). 
$$
By definition, 
\begin{eqnarray}\label{eqn-34DPDOGold2}
s_t = \tr((\alpha^t)^4 -a(\alpha^t)^2+(a-1)\alpha^t)+\tr(1-a-a^2).
\end{eqnarray}

It can be easily proved that $\ell_{1}=\ell_{n-1}=\ell_{4}=\ell_{n-4}=\ell_{2}=\ell_{n-2}=m$ 
and that the $3$-cyclotomic cosets $C_1$, $C_2$ and $C_4$ are pairwise disjoint. The desired 
conclusions on the linear span and the minimal polynomial $\m_s(x)$ then follow from 
Lemma \ref{lem-ls2} and (\ref{eqn-34DPDOGold2}). 
\end{proof}

The following theorem gives information on the code $\C_{s}$.    

\begin{theorem}\label{thm-34DO2Gold} 
Let $q=p=3$ and $m \ge 3$. 
Then the code $\C_{s}$ defined by the sequence of Lemma \ref{lem-34DPDOGold} has parameters 
$[n, n-\ls_s, d]$ and generator polynomial $\m_s(x)$, where $\m_s(x)$ and $\ls_s$ are 
given in Lemma \ref{lem-34DPDOGold}, and 
\begin{eqnarray*}
\left\{ 
\begin{array}{ll} 
d = 2   & \mbox{if } a =1,  \\
d = 3   & \mbox{if } a =0  \mbox{ $m \equiv 0 \pmod{6}$,}   \\
d \ge 4   & \mbox{if } a =0  \mbox{ $m \not\equiv 0 \pmod{6}$,}   \\
d \ge 5   & \mbox{if } a^2 \ne a \mbox{ and }  \delta(1-a-a^2)=0, \\
d = 6   & \mbox{if } a^2 \ne a \mbox{ and }  \delta(1-a-a^2) =1. 
\end{array} 
\right. 
\end{eqnarray*} 
\end{theorem} 

\begin{proof} 
The dimension of $\C_{s}$ follows from Lemma \ref{lem-34DPDOGold} and the definition of the 
code $\C_s$.  We now prove the conclusion on the minimum distance $d$ of $\C_{s}$. 

We consider the case $a=1$ first. In this case, the generator polynomial of this code $\C_{s}$ is 
$(x-1)^{\delta(1)} m_{\alpha^{-4}}(x)  m_{\alpha^{-2}}(x)$. It is easily seen that $1$, 
$\alpha^{-2}$ and $\alpha^{-4}$ are roots of $2+x^{(3^m-1)/2}=0$. Therefore, $\C_{s}$
has the codeword $2+x^{(3^m-1)/2}$ of Hamming weight 2. Hence $d=2$ when $a=1$.  

We now treat the case $a=0$.  In this case, the generator polynomial of this code is 
$\m_s(x)=(x-1)^{\delta(1)} m_{\alpha^{-4}}(x)  m_{\alpha^{-1}}(x)$. Note that 
$\bar{\m}_s(x)$ has the zeros $\alpha^3$ and $\alpha^4$. By the BCH bound the minimum weight $d$ 
in $\C_s$ is at least 3. We want to know when  $\C_s$ and $\bar{\C}_s$ have a 
codeword of weight 3. 

The code $\bar{\C}_s$ has a codeword of weight three if and only if 
there are two integers $t_1$ and $t_2$ with $1 \le t_1 \ne t_2 \le n-1$ and two elements 
$u_1$ and $u_2$ in $\{1, -1\}$ such that 
\begin{eqnarray}\label{eqn-14A1}
\left\{ \begin{array}{l} 
1 + u_1\alpha^{t_1} + u_2\alpha^{t_2}=0,  \\
1 + u_1\alpha^{4t_1} + u_2\alpha^{4t_2}=0. 
\end{array} 
\right. 
\end{eqnarray} 

Suppose now that $\bar{\C}_s$ has a codeword $1+u_1x^{t_1}+u_2x^{t_2}$ of 
weight 3.  
Combining the two equations of (\ref{eqn-14A1}) yields 
\begin{eqnarray}\label{eqn-14A2}
(u_1u_2+1) \alpha^{4t_2} + u_2\alpha^{3t_2} +  u_2\alpha^{t_2} + 1 + u_1=0 
\end{eqnarray}
and 
\begin{eqnarray}\label{eqn-14A3}
(u_1u_2+1) \alpha^{4t_1} + u_1\alpha^{3t_1} +  u_1\alpha^{t_1} + 1 + u_2=0. 
\end{eqnarray}

We now consider the first subcase that $u_1u_2=-1$ under the case that $a=0$. In 
this subcase, $\delta(1)=m \bmod{3} =0$ as $1+u_1+u_2=1 \ne 0$. In this subcase 
(\ref{eqn-14A2}) and  (\ref{eqn-14A3}) become 
\begin{eqnarray}\label{eqn-14A4}
\alpha^{3t_2} +  \alpha^{t_2} - u_1(1 + u_1)=0 
\end{eqnarray}
and 
\begin{eqnarray}\label{eqn-14A5}
\alpha^{3t_1} +  \alpha^{t_1} - u_2(1 + u_2)=0.  
\end{eqnarray}  
Due to symmetry, we assume that $(u_1, u_2)=(-1, 1)$. It follows from (\ref{eqn-14A4}) 
and (\ref{eqn-14A5}) that  
$$ 
\alpha^{2t_2}=-1 \mbox{ and } (\alpha^{t_1}-1)^2=-1. 
$$
When $m$ is odd, $\alpha^{(3^m-1)/2}=-1$ and $(3^m-1)/2$ is odd. Hence, $-1$ cannot 
be a square in $\gf(r)$. Therefore,  $\bar{\C}_s$ cannot have a codeword $1+u_1x^{t_1}+u_2x^{t_2}$  
when $a=0$ and $m$ is odd, where $u_1u_2=-1$. 
When $m$ is even, $m \equiv 0 \pmod{6}$ and $-1$ is a square in $\gf(r)$. Let $y_1 \in \gf(r)$ 
be a solution of $y^2=-1$, and define $t_2$ and $t_1$ such that 
$$ 
\alpha^{t_2}=y_1, \ \alpha^{t_1}=1+y_1. 
$$   
Then $t_1$ and $t_2$ are distinct and $1+x^{t_1}-x^{t_2}$ is indeed a codeword of weight 
three in $\bar{\C}_s$. Thus, $d=3$ when $m \equiv 0 \pmod{6}$.

We are ready to consider the second subcase that $u_1u_2=1$ under the case that $a=0$.  
In this subcase 
(\ref{eqn-14A2}) and  (\ref{eqn-14A3}) become 
\begin{eqnarray}\label{eqn-14A6}
\alpha^{4t_2}  -u_2\alpha^{3t_2} - u_2 \alpha^{t_2} - (1 + u_1)=0 
\end{eqnarray}
and 
\begin{eqnarray}\label{eqn-14A7}
\alpha^{4t_1}  -u_1\alpha^{3t_1} - u_1 \alpha^{t_1} - (1 + u_2)=0   
\end{eqnarray}  
When $(u_1, u_2)=(1, 1)$. It follows from (\ref{eqn-14A6}) 
and (\ref{eqn-14A7}) that  
$$ 
(\alpha^{t_2}-1)^4=0 \mbox{ and } (\alpha^{t_1}-1)^4=0. 
$$
Hence $\alpha^{t_2}=\alpha^{t_1}=1$. This is impossible as $\alpha$ is a generator of $\gf(r)$. 
Therefore,  $\bar{\C}_s$ cannot have a codeword $1+x^{t_1}+x^{t_2}$. 
When $(u_1, u_2)=(-1, -1)$. It follows from (\ref{eqn-14A6}) 
and (\ref{eqn-14A7}) that  
\begin{eqnarray*} 
&& \alpha^{t_2}(\alpha^{3t_2} + \alpha^{2t_2}+1)=0,  \\
&& \alpha^{t_1}(\alpha^{3t_1} + \alpha^{2t_1}+1)=0.
\end{eqnarray*} 
Note that $y^3+y^2+1=0$ if and only if 
$$ 
(y^{-1}-1)^3 +(y^{-1}-1)=0. 
$$
However, $z^3+z=0$ does not have a nonzero solution $z$ in $\gf(r)$ if $m$ is odd. 
This proves that the code $\bar{\C}_s$ cannot have a codeword $1-x^{t_1}-x^{t_2}$ 
when $m$ is odd.  
If $m$ is even, $m \equiv 0 \pmod{6}$ as $1+u_1+u_2 =1 \ne 0$. When $m \equiv 0 \pmod{6}$, 
let $z_1 \in \gf(r)$ and $z_2 \in \gf(r)$ be the two distinct solutions of $z^2=-1$. Define $t_1$ 
and $t_2$ so that 
$$ 
\alpha^{t_i}=\frac{1}{1+z_i}
$$  
for $i \in \{1,2\}$. Then $1-x^{t_1}-x^{t_2}$ is a codeword of weight three in $\bar{\C}_s$. 
This completes the proof of the conclusions on the minimum weight $d$ for the case $a=0$. 

When $a(a-1) \ne 0$, $\bar{\m}_s(x)$ has the zeros $\alpha^i$ for all $i \in \{1,2,3,4\}$ and 
the additional zero $\alpha^0$ if $\delta(1-a-a^2) =1$. The last two lower bounds on $d$ 
then follow from the BCH bound. When $a^2 \ne a$ and $\delta(1-a-a^2)=1$, the Sphere Packing 
Bound proves that $d \le 6$. We have thus $d=6$ in this case.  
\end{proof} 

\begin{remark} 
When $a=1$, the code of Theorem \ref{thm-34DO2Gold} is neither optimal nor almost optimal.  
The code is either optimal or almost optimal in all other cases. 
\end{remark} 

Now we consider the case $q=p^t$, where $p \ge 5$ or $p=3$ and $t \ge 2$. 

\begin{lemma}\label{lem-4tDPDOGold} 
Let $m \ge 2$ and $q=p^t$, where $p \ge 5$ or $p=3$ and $t\ge 2$. 
Let $s^{\infty}$ be the sequence of (\ref{eqn-sequence}), where $f(x)=D_4(x,a)=x^4-4ax^2+2a^2$.  
Then the minimal polynomial $\m_s(x)$ of  $s^{\infty}$ is given by 
\begin{eqnarray*}
\m_s(x)=\left\{ 
\begin{array}{l} 
(x-1)^{\delta(1)} m_{\alpha^{-4}}(x)  m_{\alpha^{-3}}(x)  m_{\alpha^{-1}}(x)  \mbox{ if } a =\frac{3}{2}, \\
(x-1)^{\delta(1)} m_{\alpha^{-4}}(x)  m_{\alpha^{-3}}(x)  m_{\alpha^{-2}}(x)  \mbox{ if } a =\frac{1}{2}, \\
(x-1)^{\delta(1-4a+2a^2)} \prod_{i=1}^4 m_{\alpha^{-i}}(x)  \mbox{ if } a \not\in \{\frac{3}{2}, \frac{1}{2}\},   
\end{array} 
\right. 
\end{eqnarray*} 
where $m_{\alpha^{-j}}(x)$ and the function $\delta(x)$ were defined in Section \ref{sec-notations}, 
and the linear span $\ls_s$ of $s^{\infty}$ is given by 
\begin{eqnarray*}
\ls_s=\left\{ 
\begin{array}{ll} 
\delta(1) + 3m  & \mbox{if } a \in \{\frac{3}{2}, \frac{1}{2}\}, \\
\delta(1-4a+2a^2) + 4m  & \mbox{otherwise.}  
\end{array} 
\right. 
\end{eqnarray*} 
\end{lemma} 

\begin{proof} 
Note that 
$$ 
D_4(x+1, a)=x^4 + 4x^3+(6-4a)x^2+(4-8a)x+1-4a+2a^2. 
$$
We have then 
\begin{eqnarray}\label{eqn-4tDPDOGold2}
s_t = \tr( (\alpha^t)^4 + 4(\alpha^t)^3+(6-4a)(\alpha^t)^2+(4-8a)\alpha^t) + \tr(1-4a+2a^2)
\end{eqnarray}
for all $t \ge 0$. 

Since $m \ge 2$ and $q=p^t$, where $p \ge 5$ or $p=3$ and $t\ge 2$, one can prove 
that 
$$
\ell_{1}=\ell_{n-1}=\ell_{3}=\ell_{n-3}=\ell_2=\ell_{n-2}=\ell_4=\ell_{n-4}=m
$$ 
and that the $q$-cyclotomic cosets 
$  
C_1, C_{2}, C_3, C_4   
$ 
are pairwise disjoint.  
The desired conclusions on the linear span and the minimal polynomial $\m_s(x)$ then follow from Lemma \ref{lem-ls2} 
and (\ref{eqn-4tDPDOGold2}). 
\end{proof}

The following theorem delivers to us information on the code $\C_{s}$.    

\begin{theorem}\label{thm-4tDO2Gold} 
Let $m \ge 2$ and $q=p^t$, where $p \ge 5$ or $p=3$ and $t\ge 2$. 
Then the code $\C_{s}$ defined by the sequence of Lemma \ref{lem-4tDPDOGold} has parameters 
$[n, n-\ls_s, d]$ and generator polynomial $\m_s(x)$, where $\m_s(x)$ and $\ls_s$ are 
given in Lemma \ref{lem-4tDPDOGold}, and 
\begin{eqnarray*}
\left\{ 
\begin{array}{ll} 
d \ge 3   & \mbox{if } a =\frac{3}{2}, \\
d \ge 4   & \mbox{if } a =\frac{1}{2}, \\
d \ge 5   & \mbox{if } a \not\in \{\frac{3}{2}, \frac{1}{2}\}  \mbox{ and }  \delta(1-4a+a^2)=0, \\
d = 6   & \mbox{if } a \not\in \{\frac{3}{2}, \frac{1}{2}\}  \mbox{ and }  \delta(1-4a+a^2)=1. 
\end{array} 
\right. 
\end{eqnarray*} 
\end{theorem} 

\begin{proof} 
The dimension of $\C_{s}$ follows from Lemma \ref{lem-4tDPDOGold} and the definition of the 
code $\C_s$.  The lower bounds on the minimum weight $d$ of $\C_{s}$ follow from the BCH 
bounds and the Sphere Packing Bound. The details are left to the reader. 
\end{proof} 

\begin{remark} 
Except the cases that $a \in \{\frac{3}{2}, \frac{1}{2}\}$, the code $\C_{s}$ of Theorem 
\ref{thm-4tDO2Gold} is either optimal or almost optimal. 
\end{remark}

\section{Cyclic codes from $D_5(x,a)=x^5-5ax^3+5a^2x$}\label{sec-5DOGold} 

In this section we deal with the code $\C_s$ defined by the Dickson polynomial 
$D_5(x,a)=x^5-5ax^3+5a^2x$. We have to distinguish among the three cases: $p=2$, 
$p=3$ and $p \ge 7$. The case $p=5$ was covered in Section \ref{sec-tracex}. 
So we need to consider only the remaining cases.  

We first establish the following lemma. 

\begin{lemma} \label{lem-124}
The equation $x+x^2+x^4=0$ has a nonzero solution $x \in \gf(2^m)$ if and only if 
$m \equiv 0 \pmod{3}$. 
\end{lemma} 

\begin{proof} 
Suppose that $x+x^2+x^4=0$ for some $x \in \gf(2^m)^*$. 
Then $(x+x^2+x^4)^2=x^2+x^4+x^8=0$. Combining the two 
equations yields $x+x^8=0$. Hence $x^7=1$. Since $x \ne 1$, 
this means that $\gcd(7, 2^m)=2^{\gcd(3,m)}-1=7$. Hence  
$m \equiv 0 \pmod{3}$. 

Suppose now that $m \equiv 0 \pmod{3}$. Let $m'=m/3$. Define 
$$ 
\pi(y) = \sum_{i=0}^{m'} y^{2^{3i}} 
$$
for any $y \in \gf(2^m)$. It is well known that $\tr(y)=0$ has $2^{m-1}$ 
solutions $y \in \gf(2^m)$. One of them must satisfy that $\pi(y) \ne 0$ 
as the two functions $\pi(x)$ and $\tr(x)$ are clearly different. Let $y \in \gf(2^m)$ 
such that $\tr(y)=0$ and $\pi(y) \ne 0$. Then it is easily seen that 
$\pi(y) + \pi(y)^2 + \pi(y)^4 =\tr(y)=0.$ This completes the proof.    
\end{proof} 

We first consider the cas $q=p=2$ and  prove the following lemma.  

\begin{lemma}\label{lem-52DPDOGold} 
Let $q=p=2$ and $m \ge 5$. 
Let $s^{\infty}$ be the sequence of (\ref{eqn-sequence}), where $f(x)=D_5(x,a)=x^5-5ax^3+5a^2x$.  
Then the minimal polynomial $\m_s(x)$ of  $s^{\infty}$ is given by 
\begin{eqnarray*}
\m_s(x)=  
 \left\{ 
\begin{array}{l} 
(x-1)^{\delta(1)} m_{\alpha^{-5}}(x)  \mbox{ if } a =0, \\
(x-1)^{\delta(1)} m_{\alpha^{-5}}(x)  m_{\alpha^{-3}}(x)  \mbox{ if } 1+a+a^3=0, \\
(x-1)^{\delta(1)} \prod_{i=0}^2 m_{\alpha^{-(2i+1)}}(x)   \mbox{ if } a+a^2+a^4 \ne 0 
\end{array} 
\right. 
\end{eqnarray*} 
where $m_{\alpha^{-j}}(x)$ and the function $\delta(x)$ were defined in Section \ref{sec-notations}, 
and the linear span $\ls_s$ of $s^{\infty}$ is given by 
\begin{eqnarray*}
\ls_s=\left\{ 
\begin{array}{ll} 
\delta(1) +m    &   \mbox{ if } a =0, \\
\delta(1) +2m  &   \mbox{ if } 1+a+a^3=0, \\
\delta(1) +3m  &   \mbox{ if } a+a^2+a^4 \ne 0.  
\end{array} 
\right. 
\end{eqnarray*} 
\end{lemma} 

\begin{proof} 
Note that 
$$ 
D_5(x+1, a)=x^5+x^4+ax^3+ax^2+(1+a+a^2)x+1+a+a^2. 
$$
Since $q=2$, we have then 
$$ 
\tr(D_5(x+1, a))=\tr\left(x^5+ax^3+(a^{2^{m-1}}+a+a^2)x\right)+\tr(1). 
$$
By definition, 
\begin{eqnarray}\label{eqn-52DPDOGold2}
s_t = \tr\left((\alpha^t)^5+a(\alpha^t)^3+(a^{2^{m-1}}+a+a^2)(\alpha^t)\right)+\tr(1). 
\end{eqnarray}

It can be easily proved that $\ell_{1}=\ell_{3}=\ell_{5}=m$ and that $C_1$, $C_{3}$ and 
$C_5$ are pairwise disjoint when $m \ge 5$. 
The desired conclusions on the linear span and the minimal polynomial $\m_s(x)$ then follow from Lemma \ref{lem-ls2} 
and (\ref{eqn-52DPDOGold2}). 
\end{proof}

The following theorem describes parameters of the code $\C_{s}$.    

\begin{theorem}\label{thm-52DO2Gold} 
Let $q=p=2$ and $m \ge 5$. 
Then the code $\C_{s}$ defined by the sequence of Lemma \ref{lem-52DPDOGold} has parameters 
$[n, n-\ls_s, d]$ and generator polynomial $\m_s(x)$, where $\m_s(x)$ and $\ls_s$ are 
given in Lemma \ref{lem-52DPDOGold}, and 
\begin{eqnarray*}
\left\{ 
\begin{array}{ll} 
d = 2      & \mbox{if } a =0 \mbox{ and }  \delta(1)=0 \mbox{ and } \gcd(5, n)=5, \\
d = 3      & \mbox{if } a =0 \mbox{ and }  \delta(1)=0 \mbox{ and } \gcd(5, n)=1, \\
d = 4      & \mbox{if } a =0 \mbox{ and }  \delta(1)=1, \\
d \ge 3   & \mbox{if } 1+a+a^3=0 \mbox{ and }  \delta(1)=0, \\
d \ge 4   & \mbox{if } 1+a+a^3=0 \mbox{ and }  \delta(1)=1, \\
d \ge 7   & \mbox{if } a+a^2+a^4 \ne 0 \mbox{ and }  \delta(1)=0, \\
d = 8   & \mbox{if } a+a^2+a^4 \ne 0 \mbox{ and }  \delta(1)=1. 
\end{array} 
\right. 
\end{eqnarray*} 
\end{theorem} 

\begin{proof} 
The dimension of $\C_{s}$ follows from Lemma \ref{lem-52DPDOGold} and the definition of the 
code $\C_s$.  We need to prove the conclusion on the minimum distance $d$ of $\C_{s}$. 

We consider the case $a=0$ first. Since $\alpha^5 \ne 0$, $d \ge 2$. On the other hand, 
if $\delta(1)=0$ and $\gcd(5, n)=5$, then $m$ is even and $(\alpha^{5})^{(2^m-1)/5}=1$. 
Hence $\C_s$ has the codeword $1+x^{(2^m-1)/5}$ of Hamming weight 2. Whence, $d=2$. 
If $\delta(1)=0$ and $\gcd(5, n)=1$, then $\alpha^5$ is a primitive element, the code $\C_{s}$ 
is equivalent to the Hamming code. Hence $d=3$. 
If $\delta(1)=1$, then $m$ is odd and 
$\gcd(5, 2^m-1)=1$. Hence, $\alpha^5$ is a primitive element of $\gf(2^m)$ and the code 
$\tilde{\C}_s$ generated by $\m_{\alpha^{-5}}(x)$ has minimum weight 3. Hence the even-like  
subcode $\C_s$ of  $\tilde{\C}_s$ has minimum weight 4. 

We now consider the case that $1+a+a^3=0$. By Lemma \ref{lem-124}, $m \equiv 0 \pmod{3}$. 
In this case $\m_s(x)=(x-1)^{\delta(1)} m_{\alpha^{-5}}(x)  m_{\alpha^{-3}}(x)$.  
Since $\bar{\m}_s(x)$ has the zeros $\alpha^5$ and $\alpha^6$, $d\ge 3$. If $\delta(1)=1$, 
$\C_s$ is an even-weight code. Hence $d \ge 4$.  

We finally consider the case that $1+a+a^3 \ne 0$. Note that $\bar{\m}_s(x)$ has zeros  
$\alpha^i$ for all $i \in \{1,2,3,4,5,6\}$, and the additional zero $\alpha^0$ when $\delta(1)=1$. 
The conclusions on the minimum weight $d$ in this case follow from the BCH bound.  
When $m \geq 5$ is odd and $ a+a^2+a^4 \ne 0$, the Sphere Packing Bound tells us that $d \leq 8$. 
We have then $d=8$ in the last case.  
\end{proof} 

\begin{remark} 
The code of Theorem \ref{thm-52DO2Gold} is either optimal or almost optimal. 
The code is not a BCH code when $1+a+a^3=0$, and a BCH code in the remaining cases. 
\end{remark}

We now consider the cas $(p, q)=(2, 4)$ and  prove the following lemma.  

\begin{lemma}\label{lem-452DPDOGold} 
Let $(p, q)=(2, 4)$ and $m \ge 3$. 
Let $s^{\infty}$ be the sequence of (\ref{eqn-sequence}), where $f(x)=D_5(x,a)=x^5-5ax^3+5a^2x$.  
Then the minimal polynomial $\m_s(x)$ of  $s^{\infty}$ is given by 
\begin{eqnarray*}
\m_s(x)=   
 \left\{ 
\begin{array}{ll} 
(x-1)^{\delta(1)} m_{\alpha^{-5}}(x)  & \mbox{ if } a =0, \\
(x-1)^{\delta(1)} m_{\alpha^{-5}}(x) m_{\alpha^{-3}}(x) m_{\alpha^{-2}}(x) & \mbox{ if } a=1, \\
(x-1)^{\delta(1+a+a^2)}  m_{\alpha^{-5}}(x) m_{\alpha^{-3}}(x) m_{\alpha^{-2}}(x) m_{\alpha^{-1}}(x)  
   & \mbox{ if } a+a^2 \ne 0 
\end{array} 
\right. 
\end{eqnarray*} 
where $m_{\alpha^{-j}}(x)$ and the function $\delta(x)$ were defined in Section \ref{sec-notations}, 
and the linear span $\ls_s$ of $s^{\infty}$ is given by 
\begin{eqnarray*}
\ls_s=\left\{ 
\begin{array}{ll} 
\delta(1) +m  &   \mbox{ if } a=0, \\
\delta(1) +3m  &   \mbox{ if } a=1, \\
\delta(1) +4m  &   \mbox{ if } a+a^2 \ne 0.  
\end{array} 
\right. 
\end{eqnarray*} 
\end{lemma} 

\begin{proof} 
Note that 
$$ 
D_5(x+1, a)=x^5+x^4+ax^3+ax^2+(1+a+a^2)x+1+a+a^2. 
$$
Since $q=2^2$, we have then 
\begin{eqnarray*} 
\tr(D_5(x+1, a)) = \tr\left(x^5+ax^3+ax^2+(a+a^2)x\right)+ \tr(1+a+a^2). 
\end{eqnarray*} 
By definition, 
\begin{eqnarray}\label{eqn-452DPDOGold2}
s_t = \tr\left((\alpha^t)^5+a(\alpha^t)^3+a(\alpha^t)^2+(a+a^2)(\alpha^t)\right)+ 
     \tr(1+a+a^2). 
\end{eqnarray}

It can be easily proved that $\ell_{1}=\ell_2=\ell_{3}=\ell_{5}=m$ and that $C_1$, $C_2$, $C_{3}$ and 
$C_5$ are pairwise disjoint when $m \ge 3$. 
The desired conclusions on the linear span and the minimal polynomial $\m_s(x)$ then follow from Lemma \ref{lem-ls2} 
and (\ref{eqn-452DPDOGold2}). 
\end{proof}

The following theorem supplies information on the code $\C_{s}$.    

\begin{theorem}\label{thm-452DO2Gold} 
Let $(p,q)=(2, 4)$ and $m \ge 3$. 
Then the code $\C_{s}$ defined by the sequence of Lemma \ref{lem-452DPDOGold} has parameters 
$[n, n-\ls_s, d]$ and generator polynomial $\m_s(x)$, where $\m_s(x)$ and $\ls_s$ are 
given in Lemma \ref{lem-452DPDOGold}, and 
\begin{eqnarray*}
\left\{ 
\begin{array}{ll} 
d = 2      & \mbox{if } a =0 \mbox{ and }  \delta(1)=0 \mbox{ and } \gcd(5, n)=5, \\
d = 3   & \mbox{if } a =0 \mbox{ and }  \gcd(5, n)=1, \\
d \ge 3   & \mbox{if } a=1, \\ 
d \ge 6   & \mbox{if } a+a^2 \ne 0 \mbox{ and }  \delta(1)=0, \\
d \ge 7   & \mbox{if } a+a^2 \ne 0 \mbox{ and }  \delta(1)=1.  
\end{array} 
\right. 
\end{eqnarray*} 
\end{theorem} 

\begin{proof} 
The dimension of $\C_{s}$ follows from Lemma \ref{lem-452DPDOGold} and the definition of the 
code $\C_s$.  We need to prove the conclusion on the minimum distance $d$ of $\C_{s}$. 

The proof of the lower bounds for the case $a=0$ is the same as that of Theorem \ref{thm-52DO2Gold}. 
When $a=1$, $\bar{\m}_s(x)$ has the zeros $\alpha^2$ and $\alpha^3$. Hence $d\ge 3$ when $a=1$. 

We finally consider the case that $a+a^2 \ne 0$. Note that $\bar{\m}_s(x)$ has the zeros 
$\alpha^i$ for all $i \in \{1,2,3,4,5\}$, and the additional zero $\alpha^0$ when $\delta(1)=1$. 
The conclusions on the minimum weight $d$ in this case follow from the BCH bound.   
\end{proof} 

Examples of the code of Theorem \ref{thm-452DO2Gold} are documented in arXiv:1206.4370, 
and many of them are optimal. 

\begin{Open} 
Determine the minimum distance $d$ of the code $\C_s$ in Theorem \ref{thm-452DO2Gold}. 
\end{Open}

We now consider the case $(p, q)=(2, 2^t)$, where $t \ge 3$, and  prove the following lemma.  

\begin{lemma}\label{lem-t52DPDOGold} 
Let $(p, q)=(2, 2^t)$ and $m \ge 3$, where $t \ge 3$. 
Let $s^{\infty}$ be the sequence of (\ref{eqn-sequence}), where $f(x)=D_5(x,a)=x^5-5ax^3+5a^2x$.  
Then the minimal polynomial $\m_s(x)$ of  $s^{\infty}$ is given by 
\begin{eqnarray*}
\m_s(x)=    
 \left\{ 
\begin{array}{ll} 
(x-1)^{\delta(1)} m_{\alpha^{-5}}(x)  m_{\alpha^{-4}}(x)  m_{\alpha^{-1}}(x) & \mbox{ if } a =0, \\
\prod_{i=2}^5 m_{\alpha^{-i}}(x) &  \mbox{ if } 1+a+a^2=0, \\
(x-1)^{\delta(1+a+a^2)} \prod_{i=1}^5 m_{\alpha^{-i}}(x)  & \mbox{ if } a+a^2+a^3 \ne 0, 
\end{array} 
\right. 
\end{eqnarray*} 
where $m_{\alpha^{-j}}(x)$ and the function $\delta(x)$ were defined in Section \ref{sec-notations}, 
and the linear span $\ls_s$ of $s^{\infty}$ is given by 
\begin{eqnarray*}
\ls_s=\left\{ 
\begin{array}{ll} 
\delta(1) +3m  &   \mbox{ if } a=0, \\
\delta(1) +4m  &   \mbox{ if } 1+a+a^2=0, \\
\delta(1) +5m  &   \mbox{ if } a+a^2+a^3 \ne 0.  
\end{array} 
\right. 
\end{eqnarray*} 
\end{lemma} 

\begin{proof} 
Note that 
$$ 
D_5(x+1, a)=x^5+x^4+ax^3+ax^2+(1+a+a^2)x+1+a+a^2. 
$$
Since $q=2^t$, where $t \ge 3$, we have then 
\begin{eqnarray*}
\tr(D_5(x+1, a)) = \tr\left(x^5+x^4+ax^3+ax^2+(1+a+a^2)x\right) 
                           + \tr(1+a+a^2). 
\end{eqnarray*}
By definition, 
\begin{eqnarray}\label{eqn-t52DPDOGold2}
s_t = \tr\left((\alpha^t)^5+(\alpha^t)^4+a(\alpha^t)^3+a(\alpha^t)^2+(a+a^2)\alpha^t\right)+  
      \tr(1+a+a^2). 
\end{eqnarray}

It can be easily proved that $\ell_{i}=m$ for all $1 \le i \le 5$ and that these $C_i$, where $1 \le i \le 5$, 
are pairwise disjoint. 
The desired conclusions on the linear span and the minimal polynomial $\m_s(x)$ then follow from Lemma \ref{lem-ls2} 
and (\ref{eqn-t52DPDOGold2}). 
\end{proof}

The following theorem provides information on the code $\C_{s}$.    

\begin{theorem}\label{thm-t52DO2Gold} 
Let $(p, q)=(2, 2^t)$, where $t \ge 3$. 
Then the code $\C_{s}$ defined by the sequence of Lemma \ref{lem-t52DPDOGold} has parameters 
$[n, n-\ls_s, d]$ and generator polynomial $\m_s(x)$, where $\m_s(x)$ and $\ls_s$ are 
given in Lemma \ref{lem-t52DPDOGold}, and 
\begin{eqnarray*}
\left\{ 
\begin{array}{ll} 
d \ge 3   & \mbox{if } a =0 \mbox{ and }  \delta(1)=0,  \\
d \ge 4      & \mbox{if } a =0 \mbox{ and }  \delta(1)=1, \\
d \ge 5   & \mbox{if } 1+a+a^2=0, \\
d \ge 6   & \mbox{if } a+a^2+a^3 \ne 0 \mbox{ and }  \delta(1)=0, \\
d \ge 7   & \mbox{if } a+a^2+a^3 \ne 0 \mbox{ and }  \delta(1)=1. 
\end{array} 
\right. 
\end{eqnarray*} 
\end{theorem} 

\begin{proof} 
The proof of this theorem is similar to that of Theorem \ref{thm-452DO2Gold}, and is omitted.  
\end{proof} 

\begin{Open} 
Determine the minimum distance $d$ of the code $\C_s$ in Theorem \ref{thm-t52DO2Gold}. 
\end{Open} 

Examples of the code of Theorem \ref{thm-t52DO2Gold} can be found in arXiv:1206.4370, 
and many of them are optimal. The code of Theorem \ref{thm-t52DO2Gold} is 
not a BCH code when $a=0$, and a BCH code otherwise.    

We now consider the case $q=p=3$ and  state the following lemma and theorem without proofs.  

\begin{lemma}\label{lem-352DPDOGold} 
Let $q=p=3$ and $m \ge 3$. 
Let $s^{\infty}$ be the sequence of (\ref{eqn-sequence}), where $f(x)=D_5(x,a)=x^5-5ax^3+5a^2x$.  
Then the minimal polynomial $\m_s(x)$ of  $s^{\infty}$ is given by 
\begin{eqnarray*}
\m_s(x)=  
 \left\{ 
\begin{array}{ll} 
(x-1)^{\delta(1+a+2a^2)} m_{\alpha^{-5}}(x)  m_{\alpha^{-4}}(x) m_{\alpha^{-2}}(x) 
   & \mbox{ if } a - a^6 =0, \\
(x-1)^{\delta(1+a+2a^2)} \prod_{i=2}^5 m_{\alpha^{-i}}(x)  
   & \mbox{ if } a -a^6 \ne 0, 
\end{array} 
\right. 
\end{eqnarray*} 
where $m_{\alpha^{-j}}(x)$ and the function $\delta(x)$ were defined in Section \ref{sec-notations}, 
and the linear span $\ls_s$ of $s^{\infty}$ is given by 
\begin{eqnarray*}
\ls_s=\left\{ 
\begin{array}{ll} 
\delta(1+a+2a^2) + 3m    &   \mbox{ if } a - a^6 =0, \\
\delta(1+a+2a^2) + 4m    &   \mbox{ if } a - a^6 \ne 0. 
\end{array} 
\right. 
\end{eqnarray*} 
\end{lemma} 

\begin{proof} 
The proof is similar to that of Lemma \ref{lem-t52DPDOGold}, and is omitted here.  
\end{proof}

The following theorem gives information on the code $\C_{s}$.    

\begin{theorem}\label{thm-352DO2Gold} 
Let $q=p=3$ and $m \ge 3$. 
Then the code $\C_{s}$ defined by the sequence of Lemma \ref{lem-352DPDOGold} has parameters 
$[n, n-\ls_s, d]$ and generator polynomial $\m_s(x)$, where $\m_s(x)$ and $\ls_s$ are 
given in Lemma \ref{lem-352DPDOGold}, and 
\begin{eqnarray*}
\left\{ 
\begin{array}{ll} 
d \ge 4   & \mbox{if } a-a^6=0,  \\
d \ge 7   & \mbox{if } a-a^6\ne 0 \mbox{ and }  \delta(1+a+2a^2)=0, \\
d \ge 8   & \mbox{if } a-a^6\ne 0 \mbox{ and }  \delta(1+a+2a^2)=1. 
\end{array} 
\right. 
\end{eqnarray*} 
\end{theorem} 

\begin{proof} 
The proof of this theorem is similar to that of Theorem \ref{thm-452DO2Gold}, and is omitted.  
\end{proof} 

\begin{Open} 
Determine the minimum distance $d$ of the code $\C_s$ in Theorem \ref{thm-352DO2Gold} (our experimental data 
indicates that the lower bounds are the specific values of $d$). 
\end{Open}

Examples of the code of Theorem \ref{thm-352DO2Gold} are described in arXiv:1206.4370, 
and some of them are optimal.

We now consider the case $(p, q)=(3, 3^t)$, where $t \ge 3$, and state the following lemma and theorem 
without proofs.  

\begin{lemma}\label{lem-t53DPDOGold} 
Let $(p, q)=(3, 3^t)$ and $m \ge 2$, where $t \ge 2$. 
Let $s^{\infty}$ be the sequence of (\ref{eqn-sequence}), where $f(x)=D_5(x,a)=x^5-5ax^3+5a^2x$.  
Then the minimal polynomial $\m_s(x)$ of  $s^{\infty}$ is given by 
\begin{eqnarray*}
\m_s(x)=    
 \left\{ 
\begin{array}{ll} 
(x-1)^{\delta(1)} m_{\alpha^{-5}}(x)  m_{\alpha^{-4}}(x)  m_{\alpha^{-2}}(x) m_{\alpha^{-1}}(x) 
  & \mbox{ if } 1+a =0, \\
(x-1)^{\delta(a-1)} m_{\alpha^{-5}}(x) m_{\alpha^{-4}}(x) m_{\alpha^{-3}}(x) m_{\alpha^{-2}}(x)  
  & \mbox{ if } 1+a^2=0, \\
(x-1)^{\delta(1+a+2a^2)} \prod_{i=1}^5 m_{\alpha^{-i}}(x)  
  &  \mbox{ if } (a+1)(a^2+1) \ne 0, 
\end{array} 
\right. 
\end{eqnarray*} 
where $m_{\alpha^{-j}}(x)$ and the function $\delta(x)$ were defined in Section \ref{sec-notations}, 
and the linear span $\ls_s$ of $s^{\infty}$ is given by 
\begin{eqnarray*}
\ls_s=\left\{ 
\begin{array}{ll} 
\delta(1) +4m  &   \mbox{ if } a+1=0, \\
\delta(a-1) +4m  &   \mbox{ if } a^2+1=0, \\
\delta(1+a+2a^2) +5m  &   \mbox{ if } (a+1)(a^2+1) \ne 0.  
\end{array} 
\right. 
\end{eqnarray*} 
\end{lemma} 

\begin{proof} 
The proof is similar to that of Lemma \ref{lem-t52DPDOGold}, and is omitted here.  
\end{proof}  

The following theorem supplies information on the code $\C_{s}$.    

\begin{theorem}\label{thm-t53DO2Gold} 
Let $(p, q)=(3, 3^t)$ and $m \ge 2$, where $t \ge 2$. 
Then the code $\C_{s}$ defined by the sequence of Lemma \ref{lem-t52DPDOGold} has parameters 
$[n, n-\ls_s, d]$ and generator polynomial $\m_s(x)$, where $\m_s(x)$ and $\ls_s$ are 
given in Lemma \ref{lem-t53DPDOGold}, and 
\begin{eqnarray*}
\left\{ 
\begin{array}{ll} 
d \ge 3      & \mbox{if } a =-1 \mbox{ and }  \delta(1)=0,  \\
d \ge 4      & \mbox{if } a =-1 \mbox{ and }  \delta(1)=1,  \\
d \ge 5      & \mbox{if } a^2 =-1 \mbox{ and }  \delta(a-1)=0, \\
d \ge 6      & \mbox{if } a^2 =-1 \mbox{ and }  \delta(a-1)=1, \\
d \ge 6      & \mbox{if } (a+1)(a^2+1) \ne 0 \mbox{ and }  \delta(1+a+2a^2)=0, \\
d \ge 7      & \mbox{if } (a+1)(a^2+1)  \ne 0 \mbox{ and }  \delta(1+a+2a^2)=1. 
\end{array} 
\right. 
\end{eqnarray*} 
\end{theorem} 

\begin{proof} 
The proof of this theorem is similar to that of Theorem \ref{thm-452DO2Gold}, and is omitted.  
\end{proof} 

\begin{Open} 
Determine the minimum distance $d$ of the code $\C_s$ in Theorem \ref{thm-t53DO2Gold}. 
\end{Open} 

Examples of the code of Theorem \ref{thm-t53DO2Gold} are available in  arXiv:1206.4370, 
and some of them are optimal. The code is a BCH code, except in the case that $a=-1$.

We finally consider the case $p \ge 7$, and present the following lemma and theorem without proofs.  

\begin{lemma}\label{lem-tt752DPDOGold} 
Let $p \ge 7$ and $m \geq 2$. 
Let $s^{\infty}$ be the sequence of (\ref{eqn-sequence}), where $f(x)=D_5(x,a)=x^5-5ax^3+5a^2x$.  
Then the minimal polynomial $\m_s(x)$ of  $s^{\infty}$ is given by 
\begin{eqnarray*}
\m_s(x)=    
 \left\{ 
\begin{array}{l} 
(x-1)^{\delta(1-5a+5a^2)} m_{\alpha^{-5}}(x)  m_{\alpha^{-4}}(x)  m_{\alpha^{-2}}(x) m_{\alpha^{-1}}(x) 
   \mbox{ if } a =2, \\
(x-1)^{\delta(1-5a+5a^2)} m_{\alpha^{-5}}(x) m_{\alpha^{-4}}(x) m_{\alpha^{-3}}(x) m_{\alpha^{-1}}(x) 
     \mbox{ if } a =\frac{2}{3}, \\
(x-1)^{\delta(1-5a+5a^2)} m_{\alpha^{-5}}(x) m_{\alpha^{-4}}(x) m_{\alpha^{-3}}(x) m_{\alpha^{-2}}(x)  
     \mbox{ if } a^2-3a+1=0, \\
(x-1)^{\delta(1-5a+5a^2)} \prod_{i=1}^5 m_{\alpha^{-i}}(x)   
  \mbox{ if } (a^2-3a+1)(a-2)(3a-2) \ne 0,  
\end{array} 
\right. 
\end{eqnarray*} 
where $m_{\alpha^{-j}}(x)$ and the function $\delta(x)$ were defined in Section \ref{sec-notations}, 
and the linear span $\ls_s$ of $s^{\infty}$ is given by 
\begin{eqnarray*}
\ls_s=\left\{ 
\begin{array}{l} 
\delta(1-5a+5a^2) +4m,     \mbox{ if } (a^2-3a+1)(a-2)(3a-2) =0, \\
\delta(1-5a+5a^2) +5m,  \mbox{ otherwise. }   
\end{array} 
\right. 
\end{eqnarray*} 
\end{lemma}

\begin{proof} 
The proof is similar to that of Lemma \ref{lem-t52DPDOGold}, and is omitted here.  
\end{proof}

The following theorem provides information on the code $\C_{s}$.    

\begin{theorem}\label{thm-tt752DO2Gold} 
Let $p\ge 7$ and $m \geq 2$. 
Then the code $\C_{s}$ defined by the sequence of Lemma \ref{lem-tt752DPDOGold} has parameters 
$[n, n-\ls_s, d]$ and generator polynomial $\m_s(x)$, where $\m_s(x)$ and $\ls_s$ are 
given in Lemma \ref{lem-tt752DPDOGold}, and 
\begin{eqnarray*}
\left\{ 
\begin{array}{ll} 
d \ge 3   & \mbox{if } a =2 \mbox{ and }  \delta(1-5a+5a^2)=0,  \\
d \ge 4   & \mbox{if } a =2 \mbox{ and }  \delta(1-5a+5a^2)=1, \\
d \ge 4   & \mbox{if } a =\frac{2}{3} \mbox{ and }  \delta(1-5a+5a^2)=0,  \\
d \ge 5   & \mbox{if } a =\frac{2}{3} \mbox{ and }  \delta(1-5a+5a^2)=1, \\
d \ge 5   & \mbox{if } 1-3a+a^2=0 \mbox{ and }  \delta(1-5a+5a^2)=0, \\
d \ge 6   & \mbox{if } 1-3a+a^2=0 \mbox{ and }  \delta(1-5a+5a^2)=1, \\
d \ge 6   & \mbox{if } (a^2-3a+1)(a-2)(3a-2)  \ne 0 \mbox{ and }  
               \delta(1-5a+5a^2)=0, \\
d \ge 7   & \mbox{if } (a^2-3a+1)(a-2)(3a-2)  \ne 0 \mbox{ and }  
               \delta(1-5a+5a^2)=1. 
\end{array} 
\right. 
\end{eqnarray*} 
\end{theorem} 

\begin{proof} 
The proof of this theorem is similar to that of Theorem \ref{thm-452DO2Gold}, and is omitted.  
\end{proof} 

\begin{Open} 
Determine the minimum distance $d$ of the code $\C_s$ in Theorem \ref{thm-tt752DO2Gold}. 
\end{Open}

Examples of the code of Theorem \ref{thm-tt752DO2Gold} can be found in  arXiv:1206.4370, 
and some of them are optimal. The code is a BCH code, except in the cases $a \in \{2, 2/3\}$.

\section{Cyclic codes from other $D_i(x,a)$ for $i \ge 6$} 

Parameters of cyclic codes from $D_i(x,a)$ for $i \ge 6$ could be established in a similar 
way. However, more cases are involved and the situation is getting more complicated 
when $i$ gets bigger. Examples of the code $\C_s$ from $D_7(x,a)$ and $D_{11}(x,a)$ 
can be found in arXiv:1206.4370.

\section{Cyclic codes from Dickson polynomials of the second kind}\label{sec-2ndDPcode} 

Theorems on cyclic codes from Dickson polynomials of the second kind can be developed in 
a similar way as what we did for those from Dickson polynomials of the first kind in previous 
sections.  

Experimental data indicates that the codes from the Dickson polynomials of the first kind are 
in general better than those from the Dickson polynomials of the second kind, though some 
cyclic codes from Dickson polynomials of the second kind could also be optimal or almost 
optimal.   

\section{Sets of sequences from Dickson polynomials} 

The purpose of this section is to demonstrate that optimal sets of sequences could be constructed 
with Dickson polynomials of the first kind. As an example, we consider the Dickson polynomials 
$D_3(x, a)=x^3+ax$ over $\gf(2^m)$, where $m$ is odd. With these polynomials, we define a set of binary sequences 
by 
\begin{eqnarray}\label{eqn-sequencesetd3}
\cS=\{s(a)^\infty: a \in \gf(2^m)\}, 
\end{eqnarray} 
where 
\begin{eqnarray}\label{eqn-sequenced3}
s(a)_i=\tr(D_3(1+\alpha^i, a))=\tr((\alpha^i)^3 +a \alpha^i + a +1) 
\end{eqnarray}
for all $i \geq 0$. The period of each sequence $s(a)^\infty$ is $n:=2^m-1$. 
By Lemma \ref{lem-2DPDOGold}, the linear span of $s(a)^\infty$ equals $m$ or $m+1$ if 
$a=0$, and $2m$ or $2m+1$ otherwise. 

We now prove the following property for the set $\cS$. 

\begin{lemma} 
For any two distinct elements $a$ and $b$ in $\gf(2^m)$, the two sequences $s(a)^\infty$ and 
$s(b)^\infty$ are different. Hence, $|\cS|=2^m$. 
\end{lemma} 

\begin{proof}
Note that 
$$ 
s(b)_i-s(a)_i=\tr((b-a) \alpha^i)+ \tr(b-a). 
$$ 
It then follows that the two sequences $s(a)^\infty$ and 
$s(b)^\infty$ are equal if and only if $a=b$. 
\end{proof}

We will need the following lemma \cite{KL72}. 

\begin{lemma}\label{lem-August3} 
Let $m$ be odd. For any $a \in \gf(2^m)$ and $b \in \gf(2^m)$ with $(a, b) \neq (0, 0)$, we have 
$$ 
\sum_{x \in \gf(2^m)^*} (-1)^{\tr(ax^3+bx)} \in \{-1, -1 \pm 2^{(m+1)/2} \}.  
$$ 
\end{lemma}

For any sequence $s^\infty$, the $h$-shift of $s^\infty$, denoted by $s[h]^\infty$, is defined 
by 
$$ 
s[h]_i=s_{h+i}
$$ 
for all $i \geq 0$, where $h \geq 0$ is an integer.  

Let $s^\infty$ and $t^\infty$ be two binary sequences of period $n$. The correlation value between the 
two sequences is defined by 
$$ 
\bC(s,t)=\sum_{i=0}^{n-1} (-1)^{s_i-t_i}. 
$$

Let $\cS$ be a set of binary sequence of period $n$. Then the maximum correlation value of $\cS$, 
denoted by $\bC(\cS)$, is defined by 
$$ 
\bC(\cS)=\max\left\{\max_{s \ne t, \, 0 \leq h < n} |\bC(t[h], s)|, \ 
\max_{s=t, \, 1 \leq h <n} |\bC(t[h], s)|\right\}. 
$$

We are now ready to prove the main result of this section. 

\begin{theorem} 
Let $m \geq 3$ be odd. Define $\bar{\cS}=\cS \cup \{s(\infty)^\infty\}$, where 
$$ 
s(\infty)_i=\tr(\alpha^i)
$$ 
for all $i \geq 0$, and $\cS$ was defined in (\ref{eqn-sequencesetd3}). 
We have $|\bar{\cS|}=2^m+1$ and 
$$ 
\bC(\bar{\cS})=1+2^{(m+1)/2}. 
$$
\end{theorem}

\begin{proof}
By definition, for $a \in \gf(2^m)$ and $b \in \gf(2^m)$, we have 
$$ 
s(b)_{i+h}-s(a)_i=\tr((\alpha^{3h}-1)(\alpha^i)^3 + (b \alpha^{h}-a) \alpha^i +(b-a))
$$
for all $i \geq 0$ and $h \geq 0$. Since $m$ is odd, $\gcd(3, 2^m-1)=1$. We then deduce 
that $\alpha^{3h}-1 \ne 0$ for all $1 \leq h < n$. 

For $b \in \gf(2^m)$, we have 
$$ 
s(b)_{i+h}-s(\infty)_i=\tr(\alpha^{3h} (\alpha^{i})^3) + (b \alpha^h-1) \alpha^i +b +1) 
$$
for all $i \geq 0$ and $h \geq 0$.

The desired conclusion on the maximum 
correlation value then follows from Lemma \ref{lem-August3}. 
\end{proof}

The set $\bar{\cS}$ is a modification of the Gold sequence set, and is optimal with respect to 
both the Sidelnikov and Leveinshtein bound.

\section{Concluding remarks}

In this paper, we studied the codes derived from Dickson polynomials of the first  
kind with small degrees. It is really amazing that in most cases the cyclic codes derived from the Dickson 
polynomials of small degrees within the framework of this paper are optimal or almost optimal (see arXiv:1206.4370 for examples of optimal codes).  
  
We had to treat Dickson polynomials of small degrees case by case over finite fields with different 
characteristics as we did not see any way to treat them in a single strike. The generator polynomial 
and the dimension of the codes depend heavily on the degree of the Dickson polynomials and the 
characteristic of the base field. 

It should be noted that not all cyclic codes presented in this paper are new. Some of them are  
equivalent to some known family of cyclic codes in the literature. However, it is interesting 
to show that they can be produced when Dickson polynomials of very small degrees are plugged 
into the construction approach of this paper. It is also observed that the code $\C_s$ derived 
from the Dickson polynomials of the first kind is sometimes a BCH code. However, the dimension 
and minimum distance of BCH codes are open in general, though progress on the study of primitive 
BCH codes have been made in the past 55 years.

The idea of constructing cyclic codes employed in this paper looks simple, but was proven 
to be very promising in this paper and also in  \cite{Ding13,DZ2014,TQX}. It would be nice 
if other polynomials of special forms over finite fields can be employed in this approach 
to produce more optimal and almost optimal cyclic codes.

\begin{acknowledgements}
The author is grateful to Dr. Pascale Charpin for helpful discussions on cyclic codes in 
the past years.  
\end{acknowledgements}


\end{document}